\documentclass[a4paper]{amsart}

\usepackage[utf8]{inputenc}
\usepackage[T1]{fontenc}
\usepackage{amsmath}
\usepackage{amssymb,empheq,mathtools,dsfont,enumerate}
\usepackage{physics}

\usepackage{stmaryrd}
\usepackage[mathscr]{eucal}
\usepackage{color}
\usepackage{nicefrac}

\usepackage{algorithm,float}
\usepackage[noend]{algpseudocode}
\usepackage[usenames,dvipsnames]{xcolor}
\usepackage{tikz}
\usetikzlibrary{trees,patterns,calc}
\usetikzlibrary{calc}
\usepackage{thm-restate}
\usepackage{calrsfs}
\DeclareMathAlphabet{\pazocal}{OMS}{zplm}{m}{n}
\usepackage{xifthen}
\usepackage[colorlinks]{hyperref}
\usepackage{makecell,multirow}

\usepackage{scalerel,stackengine}

\newcommand{\ie}{i.e.}
\newcommand{\eg}{e.g.}

\makeatletter
\newcommand*\defraccourci[3]{
  \expandafter\newcommand\csname #1#3\endcsname[1][]{#2{#1}}}
\newcommand*\generate[3]{
  \@for\@i:=#1\do{\expandafter\defraccourci\expandafter{\@i}{#2}{#3}}}
\makeatother

\generate{A,B,C,G,H,I,N,P,S,T,U,V,X}{\mathcal}{c} % this defines \Ac, \Bc, \Gc, ...
\generate{D,E,F}{\pazocal}{c} % this defines \Dc, \Ec

\generate{C,F,N,R,Z}{\mathbb}{}

\generate{a,b,c,k,q,s,t,u,v,w,x,y,z}{\mathbf}{v}
\renewcommand{\vec}[1]{\mathbf{#1}}

\newcommand{\zero}[1]{\vec{0}_{#1}}

\generate{A,B,C,D,E,F,G,H,J,K,L,M,N,P,Q,R,S,T,U,V,W,X,Y,Z}{\mathbf}{m}
\renewcommand{\Im}{\mathbf{I}}

\DeclareMathOperator{\Espe}{\mathbb{E}}
\newcommand{\Esp}[2][]{\Espe
\ifthenelse{\isempty{#1}}{}{_{#1}}
\ifthenelse{\isempty{#2}}{}{\left[ #2 \right]}
}

\newcommand{\ProbId}[1]{\mathbb{P}_{ideal}\left[#1\right]}
\newcommand{\ProbRe}[1]{\mathbb{P}_{real}\left[#1\right]}

\newcommand{\sumZN}[1]{\sum_{#1 \in \Z_{N}}}
\newcommand{\sumFm}[1]{\sum_{\vec{#1} \in \F_{2}^{m}}}

\DeclareMathOperator{\lsb}{lsb}

\DeclarePairedDelimiter\ceil{\lceil}{\rceil}
\DeclarePairedDelimiter\floor{\lfloor}{\rfloor}

\newcommand{\innerprod}[2]{\langle \vec{#1} , \vec{#2} \rangle}

\newcommand{\poly}{\mathsf{poly}}

\newcommand{\CSS}{\Sc^{\textup{C}}}
\newcommand{\QSS}{\Sc^{\textup{Q}}}

\newcommand{\OO}[1]{O\left( #1 \right)}
\newcommand{\OOt}[1]{\tilde{O}\left( #1 \right)}

\newcommand*{\eqdef}{\stackrel{\text{def}}{=}}

\newcommand{\HSP}{\textup{$\mathsf{HSP}$}}
\newcommand{\DHSP}{\textup{$\mathsf{DHSP}$}}
\newcommand{\DLP}{\textup{$\mathsf{DLP}$}}
\newcommand{\DCP}{\textup{$\mathsf{DCP}$}}
\renewcommand{\SS}{\textup{$\mathsf{SS}$}}
\newcommand{\uSVP}{\textup{$\mathsf{uSVP}$}}

\newtheorem{theorem}{Theorem}

\newtheorem{definition}{Definition}

\newtheorem{lemma}{Lemma}
\newtheorem{notation}{Notation}
\newtheorem{fact}{Fact}

 \makeatletter
\newcommand\bibalias[2]{%
  \@namedef{bibali@#1}{#2}%
}

\newtoks\biba@toks
\newcommand\acite[2][]{%
  \biba@toks{\cite#1}%
  \def\biba@comma{}%
  \def\biba@all{}%
  \@for\biba@one:=#2\do{%
    \@ifundefined{bibali@\biba@one}{%
      \edef\biba@all{\biba@all\biba@comma\biba@one}%
    }{%
      \PackageInfo{bibalias}{%
        Replacing citation `\biba@one' with `\@nameuse{bibali@\biba@one}'
      }%
      \edef\biba@all{\biba@all\biba@comma\@nameuse{bibali@\biba@one}}%
    }%
    \def\biba@comma{,}%
  }%
  \edef\biba@tmp{\the\biba@toks{\biba@all}}%
  \biba@tmp
}
\makeatother

\title{Time and Query Complexity Tradeoffs for the Dihedral Coset Problem}

\author{Maxime Remaud}
\address{Eviden Quantum Lab and Inria de Paris, Paris, France}
\author{André Schrottenloher}
\address{Univ Rennes, Inria, CNRS, IRISA, Rennes, France}
\author{Jean-Pierre Tillich}
\address{Inria de Paris, Paris, France}

\begin{document}
	\maketitle
	\begin{abstract}
	The Dihedral Coset Problem ($\DCP$) in $\Z_N$ has been extensively studied in quantum computing and post-quantum cryptography, as for instance, the Learning with Errors problem reduces to it. While the Ettinger-H{\o}yer algorithm is known to solve the $\DCP$ in $\OO{\log{N}}$ queries, it runs inefficiently in time $\OO{N}$. The first time-efficient algorithm was introduced (and later improved) by Kuperberg (SIAM J. Comput. 2005). These algorithms run in a \emph{subexponential} amount of time and queries $\OOt{2^{\sqrt{c_{\DCP}\log{N}}}}$, for some constant $c_{\DCP}$.

The sieving algorithms \textit{à la} Kuperberg admit many trade-offs between quantum and classical time, memory and queries. Some of these trade-offs allow the attacker to reduce the number of queries if they are particularly costly, which is notably the case in the post-quantum key-exchange CSIDH. Such optimizations have already been studied, but they typically fall into two categories: the resulting algorithm is either based on Regev's approach of reducing the $\DCP$ with quadratic queries to a subset-sum instance, or on a re-optimization of Kuperberg's sieve where the time and queries are both subexponential.

In this paper, we introduce the first algorithm to improve in the linear queries regime over the Ettinger-H{\o}yer algorithm. We then show that we can in fact interpolate between this algorithm and Kuperberg's sieve, by using the latter in a pre-processing step to create several quantum states, and solving a \emph{quantum} subset-sum instance to recover the full secret in one pass from the obtained states. This allows to interpolate smoothly between the linear queries-exponential time complexity case and the subexponential query and time complexity case, thus allowing a fine tuning of the complexity taking into account the query cost. We also give on our way a precise study of quantum subset-sum algorithms in the non-asymptotic regime.
	\end{abstract}

\section{Introduction}

\paragraph{Hidden Subgroup Problem.}
Let $G$ be a known group and $H$ be an unknown subgroup of $G$. Finding $H$ is a problem known as the Hidden Subgroup Problem ($\HSP$). To solve it, we can query a function $f$ which satisfies a certain property with respect to $H$:

\begin{definition}[$\HSP$] The hidden subgroup problem is defined as:
	\begin{itemize}
		\item Given: a function $f: G \rightarrow S$ that is constant and distinct on the left cosets of an unknown subgroup $H$ of a group $G$, $S$ being a finite set,
		\item Find: (a generating set of) $H$.
	\end{itemize}	
\end{definition}

Many problems used to construct primitives can be reduced to an $\HSP$ instance, for example the Discrete Logarithm and Shortest Vector problems. Shor's algorithm~\cite{Sho94}, which solves the $\DLP$ and breaks the RSA cryptosystem~\cite{RSA78} in polynomial time, can actually be extended to solve the $\HSP$ for an abelian group $G$. In the general case, it is well known that the problem requires only a \emph{polynomial} (in $\log_2 |G|$) number of queries to the function $f$~\cite{EHK04}. However, time-efficient quantum algorithms are only known for very specific instances, including abelian groups, and it is widely admitted that the generic problem remains difficult for quantum algorithms.

\paragraph{Dihedral Hidden Subgroup Problem ($\DHSP$).}
While the $\HSP$ in an abelian group is quantumly easy to solve, many post-quantum primitives are related to the $\HSP$ in the \emph{dihedral group}. In this case, even if the group is very close to be abelian (it has namely an abelian subgroup of index $2$) no polynomial-time algorithm is known. This is the case of cryptosystems based on the Unique Shortest Vector Problem ($\uSVP$) in lattice-based cryptography (such as \cite{AD97,Reg04b}) or on any problem that can be reduced to the $\uSVP$ (because of a chain of reductions between several problems \cite{Reg02,LM09,SSTX09}). More concretely, the security of several primitives reduces to the $\DHSP$. The most prominent example is the isogeny-based post-quantum key-exchange CSIDH~\cite{CLM+18}, which is similar to the Diffie-Hellman protocol \cite{DH76} except that it does not rely on the period-finding problem in an abelian group (which is solvable in quantum polynomial time), but on the difficulty to invert the group action. Several related constructions~\cite{AFMP20} such as the signature schemes SeaSign~\cite{FG19,DPV19} and CSI-FiSh~\cite{BKV19} also rely on the same problem. It should be noted that these isogeny-based cryptosystems are the only major contenders for which the quantum attacker enjoys more than a quadratic speedup, as opposed to the lattice- and hash-based finalists of the NIST post-quantum standardization process~\cite{nistcall,nistreport}.

As it has been shown in~\cite{BS18,Pei19,CCJR22}, a better understanding of the security of CSIDH comes from a careful analysis of quantum $\DHSP$ algorithms. This is the motivation of our work.

\paragraph{From the $\DHSP{}$ to the $\DCP{}$.}
Solving the $\HSP$ for the dihedral group of order $2N$ is known to reduce to the specific case where the hidden subgroup is $\{(0,0), (s,1)\}$, where $s \in \Z_N$, which can in turn be reduced to a problem known as the Dihedral Coset Problem (see \cite{EH99}).

\begin{definition}[$\DCP$] The dihedral coset problem is defined as:
	\begin{itemize}
		\item Given: an oracle outputting coset states of the form $\frac{1}{\sqrt{2}} \left(\ket{0}\ket{x}+\ket{1}\ket{x+s}\right)$ for random $x \in [\![0,N]\!]$,
		\item Find: $s \in [\![0,N]\!]$
	\end{itemize}
\end{definition}

While the Ettinger-H{\o}yer algorithm~\cite{EH99} solves the $\DCP{}$ with a linear number of queries ($\OO{\log{N}}$), it runs in exponential time $\OO{N}$. This algorithm basically consists in measuring $\OO{\log{N}}$ coset states and then classically looking among all possible values for the secret $s$ the one that matches the best a statistical test. It is possible to improve over this running time by reducing the resolution of the $\DCP$ to a subset-sum problem, as described by \cite{BS18,Bon19}, at the cost of squaring the query complexity. Though it remains exponential, the time complexity becomes $\OOt{N^{c_{\SS}}}$, where $c_{\SS}$ is a constant smaller than 1 that depends on the invoked subset-sum subroutine.

In a seminal work~\cite{Kup05}, Kuperberg initiated a family of \emph{sieving} algorithms which reach subexponential time complexities (at the cost of a subexponential query complexity). The idea here is to iterate a process of combining states to build new ones with a stronger and stronger structure, until building a very specific state that allows us to guess a bit of the secret when measured. The first algorithm~\cite{Kup05} requires subexponential quantum time \emph{and} space, but it was quickly followed by an algorithm of Regev \cite{Reg04} which requires only polynomial space. Later, Kuperberg proposed his second algorithm \cite{Kup13}, which generalized Regev's while improving its exponents, giving in the end a complexity in time (classical and quantum) and classical space of $\OOt{2^{\sqrt{2\log{N}}}}$ with a quantum space of $\OO{\poly{\log{N}}}$ and $\OOt{2^{\sqrt{2\log{N}}}}$ queries, which is the state of the art so far.

\paragraph{Our contributions.}
Let $n \eqdef \ceil{\log_2 N}$. We first propose a new algorithm using a linear number of queries. It is somewhat analogous to Regev's algorithm where instead of reducing the $\DCP$ to a classical subset-sum problem, it reduces the $\DCP$ to a \emph{quantum} subset-sum problem. In the first case, the algorithm makes $\OO{n}$ queries to find one bit of the secret, meaning it has to be iterated $\OO{n}$ times. With this new algorithm, which is inspired by~\cite{Reg02,SSTX09}, we only need $\OO{1}$ quantum subset-sum instances, $\ie$, $\OO{n}$ queries, to find the whole secret.

Second, we present a simple and natural method of interpolation between Kuperberg's second algorithm (which is the state of the art) and the new algorithm we mentioned above. It consists in using Kuperberg's algorithm to more or less preprocess the states given as input to our algorithm. The difficulty of solving the inherent quantum subset-sum problem instance will depend on the preprocessing step.

Finally, as a building block of our algorithms, we study quantum subset-sum algorithms when the problem to solve is partially in superposition. We show here that we can still improve over Grover's search even under the constraint of a polynomial quantum memory, using an exponential classical memory, with or without quantum access. Specifically, we show that the QRACM-based algorithm of~\cite{BBSS20} adapts to this case and reaches a complexity $\OOt{2^{0.2356n}}$. Without QRACM, we reach a quantum time $\OOt{2^{0.4165n}}$ using $\OO{2^{0.2334n}}$ bits of classical memory, improving over a previous algorithm by Helm and May~\cite{HM20}. In both cases, we also give non-asymptotic estimates of their complexity.

All together, we can summarize the complexity exponents of the different algorithms for solving the $\DCP$ in Table \ref{table:costsAlgos}, including the new one we propose.

\begin{table}[h!]
	\centering
	\caption{Costs of algorithms for finding the whole secret $s$.}
	\label{table:costsAlgos}
	\resizebox{\textwidth}{!}{\begin{tabular}{|c||c|c|c|c|} 
		 \hline
		 Algorithm				& Queries								& Classical Time						& Quantum Time							& Classical Space \\ [0.5ex] 
		 \hline\hline
		 Kuperberg II			& $\sqrt{2n} + \frac{1}{2} \log{n} +3$	& $\sqrt{2n} + \frac{1}{2} \log{n} +3$ 	& $\sqrt{2n} + \frac{1}{2} \log{n} +3$ 	& $\sqrt{2n}$ \\
		 Regev					& $2 \log{n} +3$ 				& $0.283 n$	& $2 \log{n} +3$				& $0.283 n$ \\
		 Ettinger-Hoyer			& $\log{n}+6.5$							& $n$									& $\log{n}+6.5$							& $\log{n}$ \\
		 \hline
		 Alg. \ref{algo:QuSS} w/ QRACM	& $\log{n}+3$					& $0.238 n + 12$						& $0.238 n + \frac{3}{2} \log{n} + 12$	& $0.238 n$ \\  
		 Alg. \ref{algo:QuSS} w/o QRACM	& $\log{n}+3$					& $< 0.2324 n$							& $0.418 n + \frac{3}{2} \log{n} + 15.5$& $< 0.2324 n$ \\ 
		 \hline
	\end{tabular}}
\end{table}

We propose two versions of our algorithm, one with QRACM and one without, both using polynomial quantum space. Note that our algorithm with QRACM outperforms other algorithms using a linear number of queries when we look at the complexity in classical time + quantum time.

\paragraph{Impact on CSIDH.}

Although Kuperberg's second algorithm is the one with the best time complexity for solving the $\DCP$, it is still interesting to look at algorithms that only use a linear number of queries, since for example, CSIDH cryptanalysis via the resolution of the $\DCP$ involves the use of a very expensive oracle. 

We give in Table \ref{table:costsEstimates} a few examples of complexity exponents for parameters of CSIDH.

\begin{table}[h!]
	\centering
	\caption{Complexity exponents for some parameters of CSIDH. The quantum space is polynomial in $n$.}
	\label{table:costsEstimates}
	\begin{tabular}{|c||c|c|c|c|c|c|} 
		 \hline
		 & Algorithm					& Queries & \thead{Classical \\ Time} & \thead{Quantum \\ Time}	 & \thead{Classical \\ Space} \\ [0.5ex] 
		 \hline\hline
		  \multirow{2}{*}{\shortstack{CSIDH-512 \\ ($n=256$)}}
		  	& Regev							& $19$	& $76$	& $19$	& $73$ \\
		 	& Alg. \ref{algo:QuSS} w/ QRACM	& $11$	& $73$	& $85$	& $61$ \\
		 \hline
		  \multirow{2}{*}{\shortstack{CSIDH-1024 \\ ($n=512$)}}
		  	& Regev							& $21$	& $148$	& $21$	& $145$ \\
		 	& Alg. \ref{algo:QuSS} w/ QRACM	& $12$	& $134$	& $148$	& $122$ \\
		 \hline
		  \multirow{2}{*}{\shortstack{CSIDH-1792 \\ ($n=896$)}}
		  	& Regev							& $23$	& $257$	& $23$	& $254$ \\
		 	& Alg. \ref{algo:QuSS} w/ QRACM	& $13$	& $226$	& $240$	& $214$ \\
		 \hline
		  \multirow{2}{*}{\shortstack{CSIDH-3072 \\ ($n=1536$)}}
		  	& Regev							& $25$	& $438$	& $25$	& $435$ \\
		 	& Alg. \ref{algo:QuSS} w/ QRACM	& $14$	& $378$	& $394$	& $366$ \\
		 \hline
		  \multirow{2}{*}{\shortstack{CSIDH-4096 \\ ($n=2048$)}}
		  	& Regev							& $25$	& $583$	& $25$	& $580$ \\
		 	& Alg. \ref{algo:QuSS} w/ QRACM	& $14$	& $500$	& $516$	& $488$ \\
		 \hline
	\end{tabular}
\end{table}

\paragraph{Organization of the Paper.}
In~\autoref{sec:prelim}, we give some preliminaries on sieving algorithms for the $\DCP{}$, and subset-sum algorithms that we will use as black boxes afterwards. In~\autoref{sec:dcp}, we recall the reduction from the $\DCP$ to the subset-sum problem, and introduce our new idea of using a \emph{quantum} subset-sum solver. Our interpolation between the sieving and subset-sum approaches is detailed in~\autoref{sec:interpolation}. Finally, our contributions on quantum subset-sum algorithms, and the details of the black boxes that we used throughout the paper, are provided in~\autoref{sec:quantum-subsum}.

\section{Preliminaries}\label{sec:prelim}

In this section, we cover the main principles of sieving algorithms for the $\DCP{}$, including Kuperberg's and Regev's algorithms~\cite{Kup05,Kup13,Reg04,CJS14}. We assume knowledge of the quantum circuit model, $\ie$, the $\ket{\cdot{}}$ notation of quantum states, and basic quantum operations such as CNOT, Toffoli, the Quantum Fourier Transform (QFT), \emph{etc.} 

We estimate the \emph{time} complexity of a quantum algorithm in the quantum circuit model, as a number of \emph{$n$-bit arithmetic operations}. That is, instead of counting precisely the quantum gates, we count the $n$-bit XORs, additions, subtractions, comparisons, QFTs, depending on the complexity parameter $n$.

We work with different types of memory:
\begin{itemize}
	\item quantum memory ($\ie$, qubits): some $\DCP{}$ algorithms ($\eg$, Kuperberg's first algorithm~\cite{Kup05}) need to store many coset states, which creates a subexponential quantum memory requirement;
	\item classical memory with quantum random-access (QRACM): the QRACM (or qRAM, QROM in some papers) is a specialized hardware which stores classical data and accesses this data in quantum superposition. That is, we assume that given a classical memory of $M$ bits $y_0, \ldots, y_{M-1}$, the following unitary operation: \[ \ket{x} \ket{i} \xmapsto{\mathsf{Access}} \ket{x \oplus y_i} \ket{i} \] can be implemented in time $\OO{1}$. QRACM is a very common assumption in quantum computing, and it appears in several works on the $\DCP{}$~\cite{Kup13,Pei19} but also on collision-finding~\cite{BHT98} and subset-sum algorithms~\cite{BBSS20}.
	\item classical memory without quantum random-access: the $\mathsf{Access}$ operation can be implemented in $M$  arithmetic operations using a sequential circuit. This removes the QRACM assumption, and we fall back on the basic quantum circuit model. Some algorithms using QRACM can be re-optimized in a non-trivial way when memory access is costly, and this is the case of subset-sum~\cite{HM20}.
\end{itemize}

\subsection{Phase Vectors and Kuperberg's First Algorithm}\label{sec:KupI}

We will consider in what follows that we have access to an oracle outputting \emph{phase vectors} denoted by $\ket{\psi_k}$ and defined as:
\begin{equation*}
	\ket{\psi_k} \eqdef \frac{1}{\sqrt{2}} \left( \ket{0} + \omega_N^{sk} \ket{1} \right)
\end{equation*}
where $\omega_N = \exp(2 \iota \pi/ N)$, $\iota = \sqrt{-1}$, and $k$ is a \emph{known} uniformly distributed random element of $\Z_N$. They are obtained from coset states (the input states of the $\DCP$) by applying a QFT on $\Z_N$ on the first register and then measuring this register, since we have
\begin{equation*}
	(QFT_{N} \otimes \Im) \left(\frac{1}{\sqrt{2}} \left(\ket{x}\ket{0}+\ket{x+s}\ket{1}\right)\right) = \frac{1}{\sqrt{N}} \sumZN{k} \omega_N^{kx} \ket{k} \ket{\psi_k}.
\end{equation*}
Finding $s \in [\![0,N-1]\!]$ from a collection of phase vectors $\ket{\psi_k}$ for known uniformly distributed random $k \in [\![0,N-1]\!]$ solves both the $\DCP$ and the $\DHSP$.

\paragraph{Subexponential Algorithms.}
We will now give more details on the algorithms solving the $\DCP$ in subexponential time. Until the end of this section, it can be assumed that $N=2^n$ for the sake of simplicity, but the algorithms discussed here work for any value of $N$.

Kuperberg's initial observation is that one can combine two phase vectors $\ket{\psi_p}$ and $\ket{\psi_q}$ to construct a new phase vector. Indeed, we have:
\begin{equation*}
	\ket{\psi_p,\psi_q} \xmapsto{\text{CNOT}} \frac{1}{\sqrt{2}} \left( \ket{\psi_{p+q},0} + \omega_N^{yq}\ket{\psi_{p-q},1} \right) \enspace.
\end{equation*}
A measurement of the second qubit will leave the first one either in the state $\ket{\psi_{p - q}}$, or $\ket{\psi_{p + q}}$, depending on the bit measured. With probability $1/2$, we get $\ket{\psi_{p - q}}$. By noticing on the other hand that $\ket{\psi_{N/2}} = \Hm \ket{\lsb(s)}$ ($\lsb(s)$ being the least significant bit of $s$), Kuperberg designed a quite simple algorithm which groups the phase vectors according to their last non-zero bits. They are then combined two by two using CNOT gates. Half of the time, the difference is obtained, and it contains as many zeroes as there were bits in common. The resulting phase vectors are regrouped and the process is reiterated. As proven in~\cite{Kup05}, the target state $\ket{\psi_{N/2}}$ is then obtained in subexponential time.

\subsection{Regev's Algorithm}\label{sec:Reg}

Kuperberg's first algorithm requires to store, at each time, a subexponential number of phase vectors; thus, it has subexponential quantum memory complexity. Regev~\cite{Reg04} modified the combination routine to reduce the number of qubits to polynomial, while keeping the time complexity subexponential.

The new routine combines $m$ phase vectors for a well-chosen $m$ (to minimize the overall complexity).

Let $B$ be some chosen, arbitrary value. We start with $m$ phase vectors $\ket{\psi_{k_1}}, \ldots, \ket{\psi_{k_m}}$. We tensor the vectors, $\ie$, we obtain a sum:
\[ \bigotimes_i \ket{\psi_{k_i}} =  \sum_{\vec{b} \in \{0,1\}^m} \omega_N^{ s \innerprod{b}{k} } \ket{\vec{b}} \]
We compute $\floor{\innerprod{b}{k} / B}$ into a new qubit register, and measure a value $V$. This projects the state on the vectors $\vec{b}$ such that $\floor{\innerprod{b}{k} / B} = V$. We choose $m$ and the size of $B$ such that on average two solutions $\vec{b}$ and $\vec{b'}$ occur. The state becomes proportional to:
\[ \ket{\vec{b}} + \omega_N^{ s (\innerprod{b}{k} - \innerprod{b'}{k}) } \ket{\vec{b'}} \enspace. \]
Finally, we remap $\vec{b}, \vec{b'}$ to $0,1$ respectively. We have obtained a phase vector $\ket{\psi_k}$ with a label $k = \innerprod{b}{k} - \innerprod{b'}{k} \leq B$. Then, step by step, we can make the labels decrease until we obtain the label 1. As remarked in~\cite{BS18}, we can also obtain any label whose value is invertible modulo $N$, by multiplying all initial labels by their inverse, and applying normally the algorithm afterwards. In particular, when $N$ is odd, we can obtain all powers of two. 

Regev~\cite{Reg04} and later Childs, Jao and Soukharev~\cite{CJS14} used this combination routine to get an algorithm with $\OOt{2^{\sqrt{2n \log_2 n}}}$ queries and $\OO{n}$ quantum memory.

\subsection{Kuperberg's Second Algorithm}\label{sec:KupII}

Like the two previous ones, Kuperberg's \emph{collimation sieve}~\cite{Kup13} is a hybrid quantum/classical procedure starting from the initial phase vectors, where we need to perform both quantum computations which create new vectors, and classical computations which give their description. The difference is that phase vectors are now multi-labeled:
\[ \ket{\psi_{k_1, \ldots, k_\ell}} = \frac{1}{\sqrt{\ell}} \sum_i \omega_N^{s k_i} \ket{i} \enspace. \]
In order to control these new phase vectors, we need to know the list of all their labels. These lists will become of subexponential size, although the vector itself requires only a polynomial amount of qubits. This is why the algorithm combines a polynomial quantum memory with a subexponential \emph{classical} memory.

The combination subroutine is similar to Regev's, except that it does not necessarily reduce the list of labels down to 2. Instead, the two phase vectors are combined into a new one holding a similar number of labels, as shown in~\autoref{algo:kup-comb}.

\begin{algorithm}[tb]
	\caption{Combination routine in the collimation sieve.} \label{algo:kup-comb}
	\begin{algorithmic}[1]
		\Statex \textbf{Input:} $\ket{\psi_{k_1, \ldots, k_\ell}}$, $\dots$, $\ket{\psi_{k_1', \ldots, k_{\ell'}'}}$ such that $\forall i \leq \ell, \forall j \leq \ell', k_i < 2^a, k_j' < 2^a$, the lists of the labels
		\Statex \textbf{Output:} $\ket{\psi_{v_1, \ldots, v_{\ell''}}}$ such that $\forall i, v_i < 2^{a-r}$
		\State \emph{Quantum:} Tensor the vectors: $\sum_{i \leq \ell,j \leq \ell'} \omega_N^{s (k_i + k_j')} \ket{i}\ket{j}$
		\State \emph{Quantum:} Compute the function $i,j \mapsto \floor{ (k_i + k_j') / 2^{a-r} }$ into an ancilla register
		\State \emph{Quantum:} Measure the register, obtain a value $V$. The state collapses to:
		\[ \sum_{i, j | \floor{ (k_i + k_j') / 2^{a-r} } = V} \omega_N^{s (k_i + k_j')} \ket{i}\ket{j} \]
		\State \emph{Classical:} Compute $\{ (i, j) | \floor{ (k_i + k_j') / 2^{a-r} } = V \}$, of size $\ell''$
		\State \emph{Quantum:} Apply to the state a transformation that maps the pairs $(i,j)$ to $[\![0, \ell''-1]\!]$.
		\State Return the state and the vector of corresponding labels $k_i + k_j'$.
	\end{algorithmic}
\end{algorithm}

Originally, Kuperberg uses classical memory with quantum random-access (QRACM), an approach later followed by Peikert~\cite{Pei19}. However it only improves the trade-offs with respect to the total quantum time, and it is not necessary to reach the optimal complexity. Also, the collimation procedure presented here is from later works such as~\cite{Pei19}, as it allows to easily deal with arbitrary group orders.

Without QRACM, Steps 2 and 5 in~\autoref{algo:kup-comb} require a time complexity $\OO{ \max(\ell, \ell', \ell'') }$. This is also the classical time complexity required by Step 4, assuming that the lists of labels are sorted.

\paragraph{The Algorithm as a Merging Tree.}
Starting from a certain set of multi-labeled phase vectors, we can identify them with the classical lists of their labels. The combination step operates on these lists like a purely classical list-merging algorithm, in which new lists of labels are formed from the pairs of labels satisfying a certain condition. This algorithm can be represented as a \emph{merging tree} in which all nodes are lists of labels (resp. phase vectors).

On the classical side, Kuperberg's algorithm is thus similar to Wagner's generalized birthday algorithm~\cite{Wag02}, which is a binary merging tree of depth $\sqrt{n}$. In Wagner's algorithm, the goal is to impose stronger conditions at each level which culminate in a full-zero sum. Here, the same conditions are imposed on the labels in the phase vectors. A success in the list-merging routine is equivalent to a success in the collimation routine (we obtain a phase vector with the wanted label).

The query, time and memory complexities depend on the shape of the tree. Even though the conditions are actually chosen at random at the measurement step in~\autoref{algo:kup-comb}, we can consider them chosen at random \emph{before} the combination to analyze the algorithm.

\paragraph{Example: Optimal Time.}
The optimal time complexity is obtained with a tree with $\sqrt{2 n}$ levels. It starts with lists of size 2, $\ie$, two-labeled phase vectors. At level $i$ starting from the leaves, the lists have (expected) size $2^i$, and they are merged pairwise into a list of size $2^{i + 1}$. This means that we can eliminate $2i - (i+1) = i-1$ bits. So we should use $h$ levels, where:
\[ 1 + \ldots + h-1 = n \implies h \simeq \sqrt{2n} \enspace. \]
Since each level of merging doubles the number of lists, there are in total $2^{\sqrt{2n}}$ leaves (hence queries). The (classical) cost of merging, over the whole tree, is equal to the sum of all list sizes. It is also the (quantum) cost of the relabeling operations: $\sum_i 2^{\sqrt{2 n}-i} \times 2^i = \OO{ \sqrt{2 n} 2^{\sqrt{2 n}}}$.

To compute the memory complexity, one must note that it is not required to store whole levels of the merging tree. Instead, we compute the lists (resp. the phase vectors) depth-first, and store only one node of each level at most, $\ie$ $\sqrt{2 n}$ phase vectors. For the same reason, the classical memory complexity is $\OO{2^{\sqrt{2 n}}}$.

\paragraph{Precise Analysis.}
The analysis above is only performed on average, and in practice, there is a significant variance in the list sizes. More precise analyses were performed in~\cite{Pei19,CCJR22}. It follows from them that the list size after merging should be considered smaller than the expected one by an ``adjusting factor'' $\sqrt{3 / (2\pi)}$. Furthermore, the combination may create lists that are too large, which must be discarded. The empirical analysis of Peikert~\cite{Pei19} gives a factor $(1 - \delta)$ of loss at each level, with $\delta = 0.02$.

The smaller factor in list sizes simply means that at level $i$, we will not exactly eliminate $i-1$ bits, but $i- c$ where $c = \log_2 \left( 1 + \sqrt{\frac{3}{2\pi}} \right) \simeq 0.76$. (We can control the interval size in~\autoref{algo:kup-comb} very precisely.) Thus $h$ is solution to:
\[ \sum_{i=1}^h (i-c) = n \implies \frac{h^2}{2} - ch = n \implies h \simeq c + \sqrt{2n + 4c^2} \enspace. \]

Finally, the loss at each level induces a global multiplicative factor $(1-\delta)^{-h} = 2^{-\log_2 (1-\delta) h} \simeq  2^{0.029 h}$ on the complexity. Therefore, accounting for the adjusting factor and discards, the query complexity of the sieve is:
\begin{equation}\label{eq:kupcomp}
2^{1.029 \left(0.76 + \sqrt{2n + 2.30} \right) }
\end{equation}
and the quantum time complexity multiplies this by a factor $0.76 + \sqrt{2n + 2.30}$. The difference with the exact $2^{\sqrt{2n}}$ is not negligible, but not large either. At $n = 4608$, the two exponents are respectively $99.6$ and $96$.

\paragraph{Obtaining All the Bits of the Solution.}
The analysis above applies if we want to obtain a specific label, $\eg$, the label 1. Afterwards, the algorithm can be repeated $n$ times. For a generic $N$ (not a power of 2), one typically produces all labels which are powers of 2 and uses a QFT to directly recover the secret. This is done for example in~\cite{BS18}. Peikert~\cite{Pei19} proposed a more advanced method to recover multiple bits of the secret with each phase vector.

\begin{lemma}\label{lem:sum}
	Let $\alpha > 0$ and $n$ be a positive integer. We have \[ \sum_{i=1}^n 2^{\alpha \sqrt{i}} = \OO{\sqrt{n} 2^{\alpha \sqrt{n}}}. \]
\end{lemma}
\begin{proof}
	When $i$ is a perfect square, let say $i=j^2$, we have that $2^{\alpha \sqrt{i}} = 2^{\alpha j}$. Now for any $i$ between the two perfect squares $(j-1)^2$ and $j^2$, we have the upper bound $2^{\alpha \sqrt{i}} < 2^{\alpha j}$. In order to use this, we rewrite the sum:
	\begin{align*}
		\sum_{i=1}^n 2^{\alpha \sqrt{i}} & \le \sum_{j=0}^{\ceil{\sqrt{n}} - 1} \sum_{k=j^2+1}^{(j+1)^2} 2^{\alpha \sqrt{k}} \\
		& \le \sum_{j=0}^{\ceil{\sqrt{n}} - 1} \sum_{k=j^2+1}^{(j+1)^2} 2^{\alpha (j+1)} \\
		& = \sum_{j=0}^{\ceil{\sqrt{n}} - 1} (2j+1) 2^{\alpha (j+1)}
	\end{align*}
	Using the formula for geometric series, we obtain:
	\begin{align*}
		\sum_{i=1}^n 2^{\alpha \sqrt{i}} & \le 2^{\alpha + 1} \frac{(2^\alpha - 1) \ceil{\sqrt{n}} 2^{\alpha \ceil{\sqrt{n}}} - 2^\alpha (2^{\alpha \ceil{\sqrt{n}}} - 1)}{(2^\alpha - 1)^2} + 2^\alpha \frac{2^{\alpha \ceil{\sqrt{n}}} - 1}{2^\alpha - 1} \\
		& = \frac{2^\alpha}{2^\alpha - 1} \left((2 \ceil{\sqrt{n}} + 1) 2^{\alpha \ceil{\sqrt{n}}} - \frac{2^{\alpha + 1}}{2^\alpha - 1} (2^{\alpha \ceil{\sqrt{n}}} - 1) - 1 \right) \\
		& \le \frac{2^\alpha}{2^\alpha - 1} (2 \ceil{\sqrt{n}} + 1) 2^{\alpha \ceil{\sqrt{n}}}  \enspace.
	\end{align*}
	which allows us to conclude the proof, $\alpha$ being fixed.\qed
\end{proof}

\paragraph{Obtaining Partially Collimated Labels.}
In this paper, we will consider the task of obtaining labels which, instead of reaching a prescribed $k$, match $k$ on a certain number of bits only (we can say that the phase vectors are \emph{partially collimated}), let say $i$: this complexity is of order $2^{\sqrt{2i}}$. By Lemma~\ref{lem:sum}, we can obtain a sequence of $i$ phase vectors collimated on $1, \ldots, i$ bits with a query complexity: $\sum_{j = 1}^i 2^{\sqrt{2j}} = \OO{\sqrt{i} 2^{\sqrt{2i}}}$.

\subsection{The Subset-Sum Problem}\label{subsection:subset-sum-prelim}

As we will see in~\autoref{sec:dcp}, the $\DCP$ can be reduced to the Subset-sum problem; this leads to the most query-efficient algorithms, and depending on the cost of queries, to the best optimization for some instances.

\begin{definition}[Subset-sum]
A subset-sum instance is given by $(v, \vec{k}), v \in \Z_N, \vec{k} \in \Z_N^m$ for some modulus $N$ and integer $m$. The problem is to find a vector (or all vectors) $\vec{b} \in \{0,1\}^m$ such that $\innerprod{b}{k} = v \mod N$.
\end{definition}

When $m \simeq n = \ceil{\log_2 N}$, there is one solution on average. The instance is said to be of \emph{density one}. Heuristic classical and quantum algorithms based on the \emph{representation technique}~\cite{HJ10,BCJ11} allow to solve it in exponential time in $n$. In the following, we will use these algorithms as black boxes. We first need a classical subset-sum solver. 

\begin{fact}\label{fact:ClSS} 
We have a classical algorithm $\CSS$ which, on input a subset-sum instance $(v, \vec{k})$ of density one, finds all solutions. It has a time complexity in $\OOt{2^{c_{c\SS} n}}$ where $c_{c\SS} < 1$.
\end{fact}

Here, the parameter $c_{c\SS}$ is the best asymptotic exponent that we can obtain for classical subset-sum algorithms. If there are no constraints on the memory, we can take $c_{c\SS} = 0.283$ which is the best value known at the moment~\cite{BBSS20}.

In this paper, we will also need (quantum) algorithms solving a more difficult problem, in which $\vec{k}$ is fixed, but the target $v$ is \emph{in superposition}. We will call this type of algorithm a \emph{quantum} subset-sum solver. 

\begin{fact}\label{fact:QuSS}
	We have a quantum algorithm $\QSS$ which has a complexity cost in $\OOt{2^{c_{q\SS} n}}$ (where $c_{q\SS} < 1$), which, given an error bound $\varepsilon$, given a known (classical) $\kv \in \Z_N^m$ and on input a quantum $v$, maps:
	\[ \ket{v} \ket{\bv} \mapsto \ket{v} \ket{\bv \oplus \QSS(v)} \]
	where, for a proportion at least $1-\varepsilon$ of all $v$ admitting a solution, $\QSS(v)$ is selected u.a.r. from the solutions to the subset-sum problem, $\ie$, from the set $\{ \vec{b} | \; \innerprod{b}{k} = v\}$.
\end{fact}

Notice that in the way we implement the solver, we can only guarantee that it succeeds on a large proportion of inputs (there remains some probability of error). However, it depends on some precomputations that we can redo, to obtain a heuristically independent solver which allows to reduce $\varepsilon$ and / or to ensure that we get more solutions.

Though we could implement the function $\QSS$ by running an available classical (or quantum) subset-sum algorithm, it would then require exponential amounts of qubits. Using only $\poly(n)$ qubits, we know for sure that $c_{q\SS} \leq 0.5$, because we can use Grover's algorithm to exhaustively search for a solution $\vec{b}$. This search uses $\poly(n)$ qubits only. In~\autoref{sec:quantum-subsum}, we will show that we can reach smaller values for $c_{q\SS}$, which differ depending on whether we allow QRACM or not.

\section{Reducing $\DCP$ to a Subset-sum Problem}\label{sec:dcp}

Recall that we note $n = \ceil{\log_2 N}$, where $N$ is not necessarily a power of 2. We will focus in this section on two algorithms to solve the $\DCP$: the first one (from Regev~\cite{Reg04}) uses a classical subset-sum solver and the other (ours) uses a quantum one. 

\subsection{Using a Classical Subset-sum Solver}

By reducing Regev's algorithm to a single level, as described in \cite{BS18}, we can directly produce $\lsb(s)$ from $n$ phase vectors. This is detailed in~\autoref{algo:ClSS}.

\begin{algorithm}[tb]
	\caption{Finding $\lsb(s)$ using a classical subset-sum solver $\CSS$} \label{algo:ClSS}
	\begin{algorithmic}[1]
		\Require $\ket{\psi_{k_1}}$, $\dots$, $\ket{\psi_{k_n}}$ with $\kv \eqdef (k_1 \dots k_n) \in \Z_N^n$.
		\Ensure $\lsb(s)$.
		\State Tensor the phase vectors and append a register on $\F_2^{n-1}$
			\[ \bigotimes_{i=1}^{n} \ket{\psi_{k_i}} = \frac{1}{\sqrt{2^n}} \sum_{\bv \in \F_2^n} \omega_N^{s \innerprod{b}{k}} \ket{\bv} \]
		\State Compute the inner product of $\bv$ and $\kv$ in the ancillary register 
			\[ \frac{1}{\sqrt{2^n}} \sum_{\bv \in \F_2^n} \omega_N^{s \innerprod{b}{k}} \ket{\bv} \ket{\innerprod{b}{k} \mod 2^{n-1}} \]
		\State Measure the ancillary register  \Comment{$Z$ is a normalizing constant}
			\[ \frac{1}{\sqrt{Z}} \sum_{\substack{\bv \in \F_2^n: \\ \innerprod{b}{k} = z \mod 2^{n-1}}} \omega_N^{s \innerprod{b}{k}} \ket{\bv} \ket{z} \]
		\State Search for vectors $\vec{b_i}$ such that $\innerprod{b_i}{k} = z \mod 2^{n-1}$ using $\CSS$
		\State Project the superposition onto a pair of solutions, $\eg$, $(\bv_{1},\bv_{2})$ 
			\[ \frac{1}{\sqrt{2}} \left(\omega_N^{s \innerprod{b_1}{k}} \ket{\bv_{1}} + \omega_N^{s \innerprod{b_2}{k}} \ket{\bv_{2}} \right) \]
		\State Relabel the basis states to $(\ket{0},\ket{1})$, resulting in 
			\[ \frac{\omega_N^{s \innerprod{b_1}{k}}}{\sqrt{2}} \left(\ket{0} + \omega_N^{s \innerprod{b_2-b_1}{k}} \ket{1}\right) \]
		\State Apply a Hadamard gate on the qubit, measure it and output the result.
	\end{algorithmic}
\end{algorithm}

It can be proven that in Step 4, the number of solutions is quite small but generally enough for our purpose. In Step 5, the solution vectors we want to project our superposition on are marked in an ancillary register which is then measured. Either we will get what we want, or we will end up with a superposition of the solution vectors that were not marked, in which case we start the process again with two other solution vectors. For more details, we refer to the extensive study of Regev's algorithm by Childs, Jao and Soukharev \cite{CJS14}.

The following lemma gives us the complexity of~\autoref{algo:ClSS}, derived from Regev's algorithm.

\begin{lemma}[Subsection 3.3 \cite{BS18}]\label{lem:ClSS} 
	There exists an algorithm which finds $\lsb(s)$ with $\OO{n}$ queries and quantum time and space. It has the same usage in classical time and space as the subset-sum solver $\CSS$.
\end{lemma}

\autoref{algo:ClSS} finds one bit of the secret. In order to retrieve the whole secret, we will have to repeat this procedure $n$ times. Thus, we get an algorithm using a quadratic number of calls to the oracle, exponential classical time and space because of the subset-sum solver, linear quantum space and quadratic quantum time.

It turns out that we could solve the classical subset-sum problem on the side with a quantum computer, leading to some tradeoffs described in \cite{Bon19}. But we show hereafter that we can also build an algorithm which directly uses a quantum subset-sum solver instead of having to measure the ancillary register to get a classical instance of a subset-sum problem.

\subsection{Using a Quantum Subset-sum Solver}

The main observation that led to the design of the algorithm we introduce hereafter is that on one hand, we would like to build the superposition
\begin{equation}\label{eq:goal}
	\frac{1}{\sqrt{N}} \sumZN{j} \omega_N^{sj} \ket{j}
\end{equation}
since applying the inverse QFT on $\Z_N$ on it would directly give the secret $s$, and on the other hand, we know that it would be possible, thanks to a quantum subset-sum solver, to prepare the state 
\begin{equation}\label{eq:Regev}
	\frac{1}{\sqrt{Z(\kv)}} \sumFm{b} \omega_N^{s \innerprod{b}{k}} \ket{\innerprod{b}{k} \mod N}
\end{equation}
where $Z(\kv)$ is a normalizing constant depending on $\kv$. Indeed, preparing this state is done by using Regev's trick (see \cite{Reg02,SSTX09}), $\ie$,
\begin{itemize} 
	\item[$(i)$] by tensoring $m$ phase vectors \[ \frac{1}{\sqrt{M}} \sumFm{b} \omega_N^{s \innerprod{b}{k}} \ket{\bv} \ket{\zero{n}}, \]
	\item[$(ii)$] then computing the subset-sum in the second register to get the entangled state \[ \frac{1}{\sqrt{M}} \sumFm{b} \omega_N^{s \innerprod{b}{k}} \ket{\bv} \ket{\innerprod{b}{k} \mod N}, \]
	\item[$(iii)$] and finally disentangle it thanks to a quantum subset-sum algorithm which from $\innerprod{b}{k} \mod N$ and $\kv$ (which is classical) recovers $\bv$ and subtracts it from the first register to get the state we want.
\end{itemize}
As one can see, if we could take $m = n$ and have an isomorphism between the vectors $\bv$ and the knapsack sums $\innerprod{b}{k} \mod N$, the prepared state (\ref{eq:Regev}) would be exactly the superposition (\ref{eq:goal}).

However, there would be many cases in which multiple solutions to the subset-sum problem exist. Thus we take $m < n$ and define $M \eqdef 2^m < N$. This is different from~\autoref{algo:ClSS}, where such collisions are needed. We obtain~\autoref{algo:ideal}, which uses Regev's trick with a quantum subset-sum solver in Step 3.

\begin{algorithm}[tb]
	\caption{Ideal algorithm}\label{algo:ideal}
	\begin{algorithmic}[1]
		\Require A parameter $m < n$ and phase vectors $\ket{\psi_{k_i}}$ for $i \in [\![1,m]\!]$.
		\Ensure An element $j \in \Z_N$.
		\State Tensor the $m$ phase vectors and append a register on $\Z_N$ 
			\[ \bigotimes_{i=1}^{m} \ket{\psi_{k_i}} \ket{\zero{n}} = \frac{1}{\sqrt{M}} \sumFm{b} \omega_N^{s \innerprod{b}{k}} \ket{\bv} \ket{\zero{n}} \]
		\State Compute the inner product of $\bv$ and $\kv$ in the ancillary register 
			\[ \frac{1}{\sqrt{M}} \sumFm{b} \omega_N^{s \innerprod{b}{k}} \ket{\bv} \ket{\innerprod{b}{k} \mod N} \]
		\State Uncompute $\bv$ thanks to $\kv$ and $\ket{\innerprod{b}{k} \mod N}$
			\[ \frac{1}{\sqrt{Z(\kv)}} \sumFm{b} \omega_N^{s \innerprod{b}{k}} \ket{\zero{m}} \ket{\innerprod{b}{k} \mod N} \]
		\State Apply the inverse QFT on $\Z_N$ on the second register 
			\[ \frac{1}{\sqrt{N}} \sumZN{j} \left( \frac{1}{\sqrt{Z(\kv)}} \sumFm{b} \omega_N^{(s-j) \innerprod{b}{k}} \right) \ket{\zero{m}} \ket{j} \]
		\State Measure the state and output the resulting $j$.
	\end{algorithmic}
\end{algorithm}

Despite $M$ being smaller than $N$, some cases still yield multiple solutions, and furthermore the subset-sum solver (as given by Fact \ref{fact:QuSS}) fails on some instances. This is why we distinguish between~\autoref{algo:ideal} in which we consider the quantum subset-sum solver to be ideal ($\ie$, it finds back $\bv$ from $\innerprod{b}{k}$ and $\kv$ with certainty), and the algorithm that we actually build in practice:~\autoref{algo:QuSS}.

The analysis of~\autoref{algo:QuSS} is related to the set of $\bv$ on which the quantum subset-sum solver succeeds: $\QSS(\innerprod{b}{k}) = \bv$ for a fixed $\kv$.

\begin{notation}
	Let us denote by $\Gc(\kv)$ the set of $\bv$'s that are correctly found back by $\QSS$ for a given $\kv$: \[ \Gc(\kv) \eqdef \{\bv \in \F_2^{m}: \QSS(\innerprod{b}{k}) = \bv\} \] and let $G(\kv)$ be the size of the set $\Gc(\kv)$.
\end{notation}

\begin{algorithm}[tb]
	\caption{Finding $s$ using a quantum subset-sum solver $\QSS$}\label{algo:QuSS}
	\begin{algorithmic}[1]
		\Require A parameter $m < n$ and phase vectors $\ket{\psi_{k_{i}}}$ for $i \in [\![1,m]\!]$.
		\Ensure An element $j \in \Z_N$.
		\State Tensor the phase vectors and append a register on $\Z_N$ 
			\[ \bigotimes_{i=1}^{m} \ket{\psi_{k_{i}}} \ket{\zero{n}} = \frac{1}{\sqrt{M}} \sumFm{b} \omega_N^{s \innerprod{b}{k}} \ket{\bv} \ket{\zero{n}} \]
		\State Compute the inner product of $\bv$ and $\kv$ in the ancillary register 
			\[ \frac{1}{\sqrt{M}} \sumFm{b} \omega_N^{s \innerprod{b}{k}} \ket{\bv} \ket{\innerprod{b}{k} \mod N} \]
		\State Apply $\QSS$ to uncompute $\bv$ 
			\[ \frac{1}{\sqrt{M}} \sumFm{b} \omega_N^{s \innerprod{b}{k}} \ket{\bv \oplus \QSS(\innerprod{b}{k})} \ket{\innerprod{b}{k} \mod N} \]
		\State Measure the first register. If the result is not $\zero{m}$, abort and restart with new coset states. Otherwise, we obtain 
			\[ \frac{1}{\sqrt{G(\kv)}} \sum_{\bv \in \Gc} \omega_N^{s \innerprod{b}{k}} \ket{\zero{m}} \ket{\innerprod{b}{k} \mod N} \]
		\State Apply the inverse QFT on $\Z_{N}$ on the second register 
			\[ \frac{1}{\sqrt{N}} \sumZN{j} \left( \frac{1}{\sqrt{G(\kv)}} \sum_{\bv \in \Gc} \omega_N^{(s-j) \innerprod{b}{k}} \right) \ket{\zero{m}} \ket{j} \]
		\State Measure the state and output the resulting $j$.
	\end{algorithmic}
\end{algorithm}

We apply in Step 4 a measurement in order to disentangle the superposition we have, so we can apply an inverse QFT in the same natural way as in the ideal algorithm. We show that the probability of success of the measurement ($\ie$, of measuring $0$) is good enough for our purpose when taking $m$ close to $n$. We also prove under the same assumption that the algorithm outputs the secret with good probability. All in all, these two properties lead to our main result.

\begin{theorem}\label{th:QuSS}
	There exists an algorithm which finds $s$ using $\OO{n}$ queries and the same usage in time and space as the subset-sum solver $\QSS$.
\end{theorem}

In order to analyze~\autoref{algo:QuSS} and prove Theorem \ref{th:QuSS}, we will proceed in two steps. 

\subsubsection*{Step 1.}

The first step is to give a lower bound on $\Esp[\vec{k}]{G(\vec{k})}$. This lower bound is given by estimating the number of vectors which admit more than one possible solution.

To arrive here, we first take a look at the normalization constant $Z(\kv)$ and we compute $\Esp[\vec{k}]{Z(\vec{k})}$ (the average over all choices of $\vec{k}$). This can be done by simply looking at the measurement step in~\autoref{algo:ideal}.

\begin{lemma}\label{lem:EspZ}
	We have \[ \Esp[\kv]{Z(\kv)} = M \left(1+\frac{M-1}{N}\right). \]
\end{lemma}
\begin{proof}
Fix $\kv =(k_1,\cdots,k_m)$. For all $j \in \Z_N$, the measurement in~\autoref{algo:ideal} returns $j$ with probability:
\begin{align*}
	\ProbId{j|\kv} & = \frac{1}{NZ(\kv)} \abs{\sumFm{b} \omega_{N}^{(s-j) \innerprod{b}{k}}}^2 \\
	& = \frac{1}{NZ(\kv)} \abs{\prod_{i=1}^{m} \left(1 + \omega_{N}^{(s-j) k_i} \right)}^2 	& =& \frac{1}{NZ(\kv)} \prod_{i=1}^{m} \abs{1 + \omega_{N}^{(s-j) k_i}}^2 \\
	& = \frac{1}{NZ(\kv)} \prod_{i=1}^{m} 4\cos^2{\left(\pi k_i \frac{s-j}{N} \right)}& =& \frac{M^2}{NZ(\kv)} \prod_{i=1}^{m} \cos^2{\left(\pi k_i \frac{s-j}{N} \right)} \enspace.
\end{align*}
Furthermore, we have $\sumZN{j} \ProbId{j|\kv} = 1$, so we can write:
\begin{equation}
Z(\kv) = \frac{M^2}{N} \sumZN{j} \prod_{i=1}^{m} \cos^2{\left(\pi k_i \frac{s-j}{N} \right)}
\end{equation}
	It follows that \[ \Esp[\kv]{Z(\kv)} = \frac{M^2}{N} \sumZN{j} \Esp{\prod_{i=1}^{m} \cos^2{\left(\pi k_i \frac{s-j}{N} \right)}} \]
	and since the $k_i$ are i.i.d., we have
	\begin{align*}
		\Esp[\kv]{Z(\kv)} & = \frac{M^2}{N} \left(1 + \sum_{j \in \Z_N \setminus\{s\}} \prod_{i=1}^{m} \Esp{\cos^2{\left(\pi k_i \frac{s-j}{N} \right)}}\right) \\
		& = \frac{M^2}{N} \left(1 + (N-1) \prod_{i=1}^{m} \frac{1}{2} \right) = \frac{M}{N} (N+M-1) \enspace. \qquad \square
	\end{align*}
\end{proof}

Next, we give a relation between $G(\vec{k})$ and $Z(\vec{k})$.

\begin{lemma}\label{lem:boundG}
	For any $\vec{k}$: $$G(\kv) \ge (1-\varepsilon) \left( 2M -Z(\kv) \right).$$
\end{lemma}
\begin{proof}
	Fix $\kv$. Let $\Bc(j)$ be the set of vectors whose knapsack sum is $j$: $$\Bc(j) \eqdef \{\bv \in \F_2^m \vert \; \innerprod{b}{k} = j \mod N\}$$ and let $\Cc_i$ be the set of vectors $\bv$ that have $i$ collisions: $$\Cc_i \eqdef \{\bv \in \F_2^{m} \vert \; \#\Bc(\innerprod{b}{k}) = i \}.$$ We denote by $C_i$ the size of the set $\Cc_i$.

	If we take a closer look at $Z(\kv)$, we have that
	\begin{align*}
		Z(\kv) & = \sumZN{j} \abs{\sum_{\bv \in B(j)} \omega_N^{s \innerprod{b}{k}}}^2 & = & \sumZN{j} \abs{\omega_N^{sj}}^2 \abs{\sum_{\bv \in B(j)} 1}^2\\
		& = \sumZN{j} \sum_{\bv \in B(j)} \sum_{\bv' \in B(j)} 1 & = & \sumZN{j} \sum_{\bv \in B(j)} \sum_{\bv' \in B(\innerprod{b}{k})} 1 \\
		& = \sumZN{j} \sum_{\bv \in B(j)} \#\Bc(\innerprod{b}{k}) & = & \sumFm{b} \#\Bc(\innerprod{b}{k}) \\
		& = \sum_{i \ge 1} \sum_{\bv \in \F_2^m \colon \#\Bc(\innerprod{b}{k}) = i} i & = & \sum_{i \ge 1} i C_i
	\end{align*}
	Letting $C_{> 1}$ be the number of vectors $\bv$ with at least one collision ($\ie$, for which there exists $\bv' \neq \bv$ such that they have the same knapsack sum), we have $C_{> 1} = \sum_{i > 1} C_i$. From
	\[ Z(\kv) = \sum_{i \ge 1} i C_i = C_1 + 2 \sum_{i \ge 2} C_i + \sum_{i \ge 3} (i-2) C_i \enspace, \]
	it follows that we have the lower bound:
	\[ Z(\kv) \ge C_1 + 2C_{> 1} \enspace. \]
	Injecting twice the equation $C_1 = M - C_{> 1}$ in this inequality and using the trivial bound $G(\kv) \ge (1-\varepsilon) C_1$, we conclude the proof. \qed
\end{proof}

From Lemma~\ref{lem:EspZ} and~\ref{lem:boundG}, we immediately deduce:

\begin{lemma}\label{lem:boundEspG}
	\[ \Esp[\vec{k}]{G(\vec{k})} \ge (1-\varepsilon) M \left(1-\frac{M-1}{N}\right). \]
\end{lemma}

\subsubsection*{Step 2.}
The second step in our proof computes the probability of success of the ``real'' algorithm by relating it to $\Esp[\vec{k}]{G(\vec{k})}$.

\begin{lemma}\label{lem:probaRe}
	\autoref{algo:QuSS} outputs the secret $s$ with probability $ \geq (1-\varepsilon) \frac{M(N-M+1)}{N^2}$.
\end{lemma}

\begin{proof}
We compute the probability of measuring $j \in \Z_N$ at the end of~\autoref{algo:QuSS}. In particular, we have for $s$
\[ \ProbRe{s|\kv} = \frac{1}{N G(\vec{k})} \abs{\sum_{\bv \in \Gc} \omega_{N}^{0}}^2 = \frac{G(\vec{k})}{N} \]
We have by Lemma \ref{lem:boundEspG} that $\Esp{G(\vec{k})} \ge (1-\varepsilon) M \left(1-\frac{M-1}{N}\right)$. We finish the proof by observing that $\ProbRe{s}= \Esp{\ProbRe{s|\kv}} \geq (1-\varepsilon) \frac{M(N-M+1)}{N^2}$.\qed
\end{proof}

Finally, we can prove Theorem \ref{th:QuSS}.

\begin{proof}
	Step 4 of~\autoref{algo:QuSS} succeeds with average probability $\frac{\Esp{G(\vec{k})}}{M}$ which is greater than $(1-\varepsilon)\frac{N-M+1}{N}$ (by Lemma~\ref{lem:boundEspG}). The final measurement of the algorithm outputs the secret with probability $\geq (1-\varepsilon) \frac{M(N-M+1)}{N^2}$ (by Lemma~\ref{lem:probaRe}). We will thus have to repeat the algorithm an expected number smaller than $\frac{N^3}{(1-\varepsilon)^2 M (N-M+1)^2}$ times. By letting $m$ be equal to $n-1$, we obtain that the algorithm will have to be repeated less than $8/(1-\varepsilon)^2$ times. Thus, we can conclude that our algorithm needs $\OO{n}$ queries and has complexity costs identical to the ones of the subset-sum solver, since the subset-sum resolution is the only exponential step of the algorithm.\qed
\end{proof}

\section{Interpolation algorithm}\label{sec:interpolation}

If we take a look at the ideal algorithm and consider that there is no collision, we can see that we would like $2^m$ to be as close to $N$ as possible in order for the sum
\begin{equation*}
	\frac{1}{\sqrt{M}} \sumFm{b} \omega_N^{(s-j) \innerprod{b}{k}}
\end{equation*}
to contain as many as possible elements of the sum
\begin{equation*}
	\frac{1}{\sqrt{N}} \sum_{\bv \in \F_{2}^{n}} \omega_N^{(s-j) \innerprod{b}{k}}.
\end{equation*}
In the mean time, it is clear that the closest $M$ gets to $N$, the more likely collisions $\innerprod{b_1}{k} = \innerprod{b_2}{k}$ for $\bv_1 \neq \bv_2$ become. We thus have to find a compromise on the value $m$ or more interestingly play with the values $k_i$ used in the algorithm, to avoid  collisions  and to simplify the resolution of the subset-sum problem.

In fact, we can reduce the size of the subset-sum problem we have to solve by pre-processing the states to get values of $k_i$ that will allow us to solve the subset-sum problem on some bits  by Gaussian elimination. Constructing these $k_i$'s  can be achieved by Kuperberg's second algorithm (or any improvement). Given a threshold parameter $t \in [\![1,m]\!]$, we can consider the following configuration for the $k_i$ to use as inputs in~\autoref{algo:QuSS} (dots represent unknown bits and the $i$-th bit of the $j$-th row is the $j$-th bit of the binary expansion of $k_i$):

\begin{equation}\label{mat:config}
\bordermatrix{~			& 1		& 2			& \cdots& m-t		& m-t+1		& \cdots& n			\cr
              k_1		& 1		& \bullet	& \cdots& \bullet	& \bullet	& \cdots& \bullet	\cr
              k_2		& 0		& 1			& \ddots& \bullet	& \bullet	& \cdots& \bullet	\cr
              \vdots	& \vdots& \vdots	& \ddots& \ddots	& \vdots	& \cdots& \bullet	\cr
              k_{m-t}	& 0		& 0			& \cdots& 1			& \bullet	& \cdots& \bullet	\cr
              k_{m-t+1}	& 0		& 0			& \cdots& 0			& \bullet	& \cdots& \bullet	\cr
              \vdots	& \vdots& \vdots	& 		& \vdots	& \vdots	& \cdots& \bullet	\cr
              k_{m}		& 0		& 0			& \cdots& 0			& \bullet	& \cdots& \bullet	\cr}
\end{equation}

In~\autoref{algo:QuSS}, it turns out that we can keep a good probability of finding the secret $s$ by letting $m$ be equal to $n-1$ so that is what we will assume afterwards.

To build phase vectors that satisfy the configuration in \eqref{mat:config}, we will approximately have to query the oracle
\begin{equation*}
	\sum_{i=1}^{n-t} 2^{\sqrt{c_{\DCP} i}} + t 2^{\sqrt{c_{\DCP}(n-t)}}
\end{equation*}
leading to a query and time complexities of $\OO{(\sqrt{n-t} + t) 2^{\sqrt{c_{\DCP}(n-t)}}}$ (by Lemma \ref{lem:sum}), where $c_{\DCP}$ is the constant of the algorithm used to construct the states ($c_{\DCP} = 2$ for Kuperberg's second algorithm). For the subset-sum problem, solving it on the first $n-t$ bits is easy (thanks to a Gaussian elimination), the difficulty comes from the last $t$ bits, leading to a complexity in $\OO{2^{c_{q\SS} t}}$ time, where $c_{q\SS}$ is the complexity exponent of the quantum subset-sum solver. This parameter $t$ can be used in a natural way to obtain an interpolation algorithm, since it allows to obtain a tradeoff between the preparation of the states and the resolution of the problem (which amounts to solving a quantum subset-sum problem). 

We can now give an interpolation algorithm derived from~\autoref{algo:QuSS}. We note that letting $q$ be the query complexity exponent, it is possible to determine $t$ from $n$ and the value $q$ we can afford. Using Kuperberg's second algorithm (or any improvement) to compute suitable phase vectors as described before and then giving them as inputs to~\autoref{algo:QuSS}, we can retrieve the secret $s$ as described by~\autoref{algo:interpol} with the complexities given by~\autoref{th:interpol}.

\begin{algorithm}[htb]
	\caption{Interpolation algorithm (using a quantum $\SS$ solver)}\label{algo:interpol}
	\begin{algorithmic}[1]
		\Require $q$ such that $2^q$ is the number of queries we are allowed to do.
		\Ensure The secret $s$.
		\State Use Kuperberg's second algorithm (or any improvement) to create states $\ket{\psi_{k_{i}}}$ for $i \in [\![1,m]\!]$ satisfying the configuration represented by Matrix \eqref{mat:config}, where $t \approx \frac{q}{c_{\SS}}$.
		\State Apply~\autoref{algo:QuSS} on these $m$ states to obtain a value $j \in \Z_N$.
		\State Check if $j$ is the secret. If not, return to Step 1. Otherwise, output $j$.
	\end{algorithmic}
\end{algorithm}

\begin{restatable}{theorem}{thQumain} \label{th:interpol}
	Let $t \in [\![1, m]\!]$.~\autoref{algo:interpol} finds $s$ with $\OO{(\sqrt{n-t} + t) 2^{\sqrt{c_{\DCP}(n-t)}}}$ queries in $\OO{(\sqrt{n-t} + t) 2^{\sqrt{c_{\DCP}(n-t)}} + 2^{c_{q\SS}t}}$ quantum time, classical space $\OO{2^{\sqrt{c_{\DCP}(n-t)}} + 2^{c_{q\SS}t}}$ and $\OO{\poly(n)}$ quantum space.
\end{restatable}

We notice that when $t=m$, the $k_i$ are kept random and we have to solve the ``full rank'' subset-sum, matching with~\autoref{algo:QuSS}. On the other side, when $t = 1$, we fall back on Kuperberg's second algorithm since we have in this case to construct a collection of states divisible by all the successive powers of 2, see~\autoref{sec:KupII}. Finally, when $1 < t < m$, we have new algorithms working for any number of queries between $\OO{n}$ and $\OOt{2^{\sqrt{c_{\DCP}n}}}$.

\section{Quantum Subset-sum Algorithms}\label{sec:quantum-subsum}

In this section, we consider quantum algorithms solving the \emph{quantum} subset-sum problem introduced in~\autoref{subsection:subset-sum-prelim}. We give both asymptotic complexities and numerical estimates.

Recall that we consider a subset-sum instance $(v, \vec{k}), \vec{k} \in \Z_N^m$, where $v$ is \emph{in superposition}, and $\vec{k}$ will remain fixed. The problem is to find $\vec{b}$ such that $\innerprod{b}{k} = v \mod N$ for a given (fixed) modulus $N$. For a given $v$, if there are many solutions, we want to find one selected uniformly at random (under heuristics). If we want all solutions, then we can run multiple instances of the solver (we will have to redo the pre-computations that we define below). A given solver, defined for a specific $\vec{k}$, is expected to work only for some (large) proportion $1-\varepsilon$ of $v$. We can check whether the output is a solution or not and measure the obtained bit to collapse on the cases of success.

\subsection{Algorithms Based on Representations}

The best algorithms to solve the subset-sum problem with density one are list-merging algorithms using the \emph{representation technique}~\cite{HJ10,BCJ11}. The best asymptotic complexities (both classical and quantum) are given in~\cite{BBSS20}. We detail the representation framework following the depiction given in~\cite{BBSS20}. To ease the description, we start with the case $v = 0$, $\ie$, the \emph{homogeneous} case, and we will show below how to extend it easily to $v \neq 0$.

\paragraph{Guessed Weight.}
We assume that the solution $\vec{b}$ is of weight $\ceil{m/2}$. This is true only with probability: $p_m := 2^{-m} {m \choose \ceil{m/2}}= 1/ \mathsf{poly}(m)$. If not, we re-randomize the subset-sum instance by multiplying $\vec{b}$ by a random invertible matrix. Thus if we manage to solve an instance of weight $\ceil{m/2}$, the total complexity to solve any instance will introduce a multiplicative factor $\frac{1}{p_m}$ that we will have to estimate.

\paragraph{Distributions.}
We consider \emph{distributions} of vectors having certain relative weights: $D^m[\alpha] \subseteq \{0,1\}^m$ is the set of vectors having weight $\alpha m$. The basic idea of representations is to write the solution $\vec{b}$ as a sum of vectors of smaller relative weights, $\eg$, $\vec{b} = \vec{b}_1 + \vec{b}_2$ where $\vec{b}_1 \in D^m[\alpha_1], \vec{b}_2 \in D^m[\alpha_2]$ and $\alpha_1 + \alpha_2 = \frac{1}{2}$. In this paper, we consider only representations with coefficients 0 or 1. Extended representations can be considered, using more coefficients (which have to cancel out each other). However, the advantage of using extended representations becomes quickly insignificant in practice. It is also harder to compute the number of representations, or the filtering probabilities that we define below.

\paragraph{Merging Tree.}
A subset-sum algorithm is defined by a \emph{merging tree}. A node in this tree is a \emph{list} $L[\ell, \alpha, c]$, which represents a set of vectors drawn from $\{0,1\}^m$ under several conditions: 1. the size of the list is $2^{m \ell}$; 2. the vectors are sampled  u.a.r. from a prescribed distribution $D^m[\alpha]$; 3. the vectors satisfy a \emph{modular condition} of $c m$ bits. With $v = 0$, the following condition can be used: $\innerprod{e}{k} \mod N \in [-N/2^{cm} ; N/2^{cm}]$ for some number $c$. More generally, the modular conditions can be chosen arbitrarily, as long as they remain compatible with the target $v$.

Once the tree structure is chosen, its parameters are optimized under several constraints. First, the lists have a certain maximal size. A distribution $D^m[\alpha]$ has size ${m \choose \alpha m}$, which is asymptotically estimated as $\simeq 2^{h(\alpha)m}$ where $h(x) := - x \log_2 x - (1-x) \log_2 (1-x)$ is the Hamming entropy. This creates the constraint $\ell \leq h(\alpha) - c$. Second, we expect the root list to contain the solution of the problem, $\ie$, $\ell = 0$ (one element), $\alpha = \frac{1}{2}$ and $c = 1$. Finally, each non-leaf list $L$ has its parameters determined by its two children $L_1, L_2$. Indeed, it is obtained via the \emph{merging-filtering} operation which selects, among all pairs of vectors $(\vec{e}_1, \vec{e}_2) \in L_1 \times L_2$, the pairs such that: $\vec{e}_1 + \vec{e}_2$ satisfies the modular condition (merging) and satisfies the weight condition (filtering). The parameters are:
\begin{equation}
	\begin{cases}
		\alpha = \alpha_1 + \alpha_2 \text{ (increasing weights)} \\
		\ell = \ell_1 + \ell_2 - (c - \min(c_1, c_2)) - \mathsf{pf}(\alpha_1, \alpha_2)
	\end{cases}
\end{equation}

Here, $\mathsf{pf}$ is the probability that two vectors chosen u.a.r. in their respective distributions will not have colliding 1s.

\begin{lemma}[Lemma 1 in~\cite{BBSS20}]\label{lemma:pf}
	Let $\vec{e_1}, \vec{e_2}$ be drawn u.a.r. from $D^m[\alpha_1], D^m[\alpha_2]$ with $\alpha_1 + \alpha_2 \leq 1$. The probability that $\vec{e}_1 + \vec{e}_2 \in D^m[\alpha_1 + \alpha_2]$ is equal to:
	\[ \mathsf{PF}(\alpha_1, \alpha_2, m) :=  {m - \alpha_1 m \choose \alpha_2 m }  / {m \choose \alpha_2 m} \simeq 2^{m  \mathsf{pf}(\alpha_1, \alpha_2) } \]
	where $\mathsf{pf}(\alpha_1, \alpha_2) := h\left(\frac{1-\alpha_2}{\alpha_1}\right) \alpha_1 - h(\alpha_1) \enspace.$
\end{lemma}

\paragraph{Classical Computation of the Tree.}
To any correctly parameterized merging tree corresponds a classical subset-sum algorithm that runs as follows: it creates the leaf lists by sampling their distributions at random. It then builds the parent lists by \emph{merging-filtering} steps. The \emph{merging} operation is efficient, since elements can be ordered according to the modular condition to be satisfied.

\begin{lemma}[Lemma 2 in~\cite{BBSS20}]
	Let $L_1, L_2$ be two sorted lists stored in classical memory with random access. In $\log_2$, relatively to $m$, the parent list $L$ can be built in time: $\max( \min(\ell_1, \ell_2), \ell_1 + \ell_2 -(c- \min(c_1, c_2)))$ and in memory $\max(\ell_1, \ell_2, \ell)$.
\end{lemma}

\paragraph{Quantum Computation of the Tree.}
While the more advanced quantum subset-sum algorithms use quantum walks~\cite{BJLM13,HM18,BBSS20}, we want to focus here on algorithms using few qubits, which at the moment, rely only on quantum merging with Grover search. They replace the classical merging operation by the following.

\begin{lemma}[Lemma 4 in~\cite{BBSS20}]\label{lemma:quantum-merging}
	Let $L_2$ be a sorted list stored in QRACM. Assume given a unitary $U$ that produces, in time $t_{L_1}$, a uniform superposition of elements of $L_1$. Then there exists a unitary $U'$ that produces a uniform superposition of elements of $L$, in time $\OO{ \frac{t_{L_1}}{\sqrt{ \mathsf{pf}(\alpha_1, \alpha_2) }} \max( \sqrt{2^{cm} / |L_2|},1 ) }$.
\end{lemma}

Since the goal is only to sample u.a.r. from the root list, only half of the lists in the tree need actually to be stored in QRACM. The others are sampled using the unitary operators given by Lemma~\ref{lemma:quantum-merging}. In short, the obtained subset-sum algorithm is a sequence of Grover searches which use existing lists stored in memory to sample elements in new lists with more constraints.

\paragraph{Heuristics.}
The standard subset-sum heuristic assumes that the elements of all lists in the tree (not only the leaf lists) behave as if they were uniformly sampled from the set of vectors of right weight, satisfying the modular condition. This heuristic ensures that the list sizes are very close to their average: for each $L$ obtained by merging and filtering $L_1[\ell_1, \alpha_1, c_1]$ and $L_2[\ell_2, \alpha_2, c_2]$, we have:
\[ |L| \simeq \frac{|L_1| |L_2|}{ 2^{m(c - \min(c_1,c_2))} \mathsf{PF}(\alpha_1, \alpha_2, m) } \enspace,  \]
where the approximation is exact down to a factor 2. This is true with overwhelming probability for all lists of large expected size via Chernoff-Hoeffding bounds, and even if the root list is of expected size 1, the probability that it actually ends up empty is smaller than $e^{-0.5} \simeq 0.61$.

\subsection{From Asymptotic to Exact Optimizations}

As the time and memory complexities of a subset-sum algorithm are determined by its merging tree, we seek to select a tree which minimizes these parameters. Given a certain subset-sum problem, we first select a tree shape. As an example, the best subset-sum algorithm with low qubits (using QRACM) is the ``quantum HGJ'' algorithm of~\cite{BBSS20}, whose structure is reproduced in~Figure~\ref{fig:hgj-quantum}. At level 3, it splits the vectors into two halves, and merges without filtering. While all lists are obtained via quantum merging/filtering, the main computation is performed after obtaining $L_1^3, L_1^2, L_1^1$, where the main branch is explored using Grover's algorithm: we search through the lists $L_0^3, L_0^2 L_0^1$ without representing them in memory. The quadratic speedup of Grover search makes the tree unbalanced, which is reflected on the naming of its parameters in~Figure~\ref{fig:hgj-quantum}.

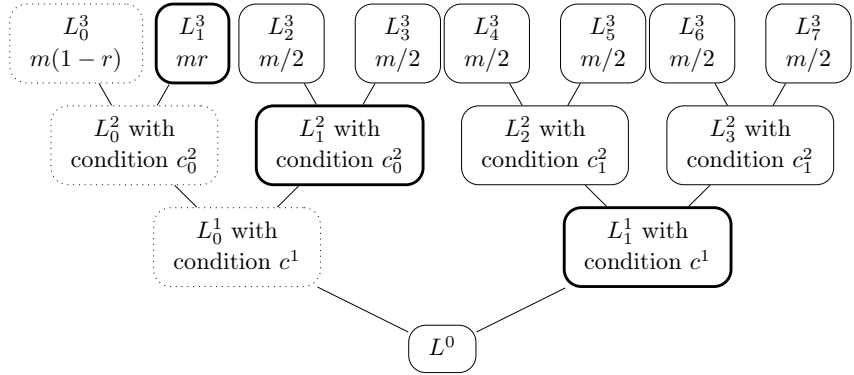
\begin{figure}[htb]
	\centering
	\scalebox{0.9}{\begin{tikzpicture}[grow=up,nodes={draw,rectangle,rounded corners=.25cm,->}, level 1/.style={sibling distance=60mm}, level 2/.style={sibling distance=30mm}, level 3/.style={sibling distance=17mm}]
	\node{ \begin{tabular}{c} $L^0$ \end{tabular}}
	    child { node[very thick] {\begin{tabular}{c} $L_1^1$ with \\ condition $c^1$ \end{tabular} }
	    	child { node {\begin{tabular}{c} $L_3^2 $ with \\ condition $c_1^2$ \end{tabular} }
	    		child { node {\begin{tabular}{c} $L_7^3 $ \\ $m/2$\end{tabular} }}
				child { node {\begin{tabular}{c} $L_6^3 $ \\ $m/2$\end{tabular} }}}
	    	child { node {\begin{tabular}{c} $L_2^2 $ with \\ condition $c_1^2$ \end{tabular} }
				child { node {\begin{tabular}{c} $L_5^3 $ \\ $m/2$\end{tabular} }}
				child { node {\begin{tabular}{c} $L_4^3 $ \\ $m/2$\end{tabular} }} }}
	    child { node[dotted] {\begin{tabular}{c} $L_0^1$ with \\ condition $c^1$ \end{tabular} }
	    	child { node[very thick] {\begin{tabular}{c} $L_1^2 $ with \\ condition $c_0^2$ \end{tabular} }
	    		child { node {\begin{tabular}{c} $L_3^3 $\\ $m/2$\end{tabular} }}
				child { node {\begin{tabular}{c} $L_2^3 $\\ $m/2$ \end{tabular} }} 
				 }
	    	child { node[dotted] {\begin{tabular}{c} $L_0^2 $ with \\ condition $c_0^2$ \end{tabular} }
				child { node[very thick]         {\begin{tabular}{c} $L_1^3 $\\ $m r$ \end{tabular} }}			
				child { node[dotted] {\begin{tabular}{c} $L_0^3 $\\ $m(1-r)$ \end{tabular} }} 
				 }};
	\end{tikzpicture}}
	\caption{Quantum HGJ algorithm. Dotted lists are search spaces (they are not stored). Bold lists are stored in QRACM. The first level uses a left-right split of vectors, without filtering.}
	\label{fig:hgj-quantum}
\end{figure}

The \emph{asymptotic} time complexity of the algorithms has the form $\OOt{2^{\beta m}}$, and is the result of summing together the costs of all merging steps. Through the approximation of binomial coefficients, the list sizes are approximated in $\log_2$ and relatively to $m$. The parameters (relative weights, modular conditions and sizes) are numerically optimized. The optimization of~Figure~\ref{fig:hgj-quantum} in~\cite{BBSS20} yields the complexity $\OOt{2^{0.2356 m}}$.

In this paper, we also perform non-asymptotic optimizations for a given $m$. Since we use only $\{0,1\}$-representations, the filtering probability is well known and has a simple expression (Lemma~\ref{lemma:pf}). Since the binomial coefficients can be extended as functions of $\R^2$, we can perform an exact numerical optimization of list sizes for a given $m$. Afterwards, the numbers obtained are rounded, in particular the weights of representations, and we take the point which gives us the best results: smallest complexity and biggest average size for $L^0$.

\paragraph{Example.}
Let us take $n = \log_2 N = 256$, $m = n-1 = 255$, and the structure of~\autoref{fig:hgj-quantum}. We adapt the optimization code of~\cite{BBSS20} by taking the exact exponents (not relative to $n$) and optimize numerically under the constraint $|L^0| = 2^2$ (to ensure that there are solutions). The asymptotic formula would give $2^{0.2356 n} \simeq 2^{60.31}$. Numerical optimization gives us a time $2^{63.81}$, but this admits non-integer parameters and it is only the \emph{maximum} between all steps. By rounding the parameters well, we obtain~\autoref{fig:hgj-numerical}.

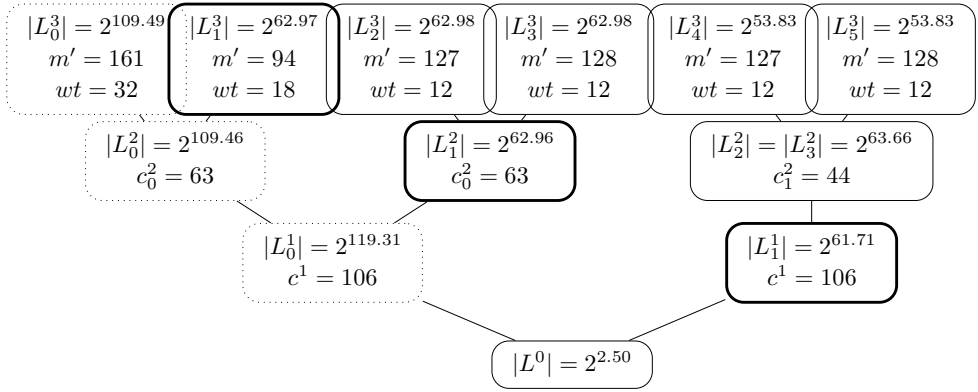
\begin{figure}[htb]
	\centering
	\scalebox{0.9}{\begin{tikzpicture}[grow=up,nodes={draw,rectangle,rounded corners=.25cm,->}, level 1/.style={sibling distance=70mm}, level 2/.style={sibling distance=46mm}, level 3/.style={sibling distance=23mm}]
	\node{ \begin{tabular}{c} $|L^0| = 2^{2.50}$ \end{tabular}}
	    child { node[very thick] {\begin{tabular}{c} $|L_1^1| = 2^{61.71}$ \\ $c^1 = 106$ \end{tabular} }
	    	child { node {\begin{tabular}{c} $|L_2^2| = |L_3^2| = 2^{63.66}$ \\ $c_1^2 = 44$ \end{tabular} }
				child { node {\begin{tabular}{c} $|L_5^3| = 2^{53.83}$ \\ $m' = 128$ \\ $wt = 12$ \end{tabular} }}
				child { node {\begin{tabular}{c} $|L_4^3| = 2^{53.83} $ \\ $m' = 127$ \\ $wt = 12$ \end{tabular} }} }}
	    child { node[dotted] {\begin{tabular}{c} $|L_0^1| = 2^{119.31}$ \\ $c^1 = 106$ \end{tabular} }
	    	child { node[very thick] {\begin{tabular}{c} $|L_1^2| = 2^{62.96} $ \\ $c_0^2 = 63$ \end{tabular} }
	    		child { node {\begin{tabular}{c} $|L_3^3| = 2^{62.98} $\\ $m' = 128$ \\ $wt = 12$ \end{tabular} }}
				child { node {\begin{tabular}{c} $|L_2^3| = 2^{62.98} $\\ $m' = 127$ \\ $wt = 12$ \end{tabular} }} 
				 }
	    	child { node[dotted] {\begin{tabular}{c} $|L_0^2| = 2^{109.46} $ \\ $c_0^2 = 63$ \end{tabular} }
				child { node[very thick]         {\begin{tabular}{c} $|L_1^3| = 2^{62.97} $\\ $m' = 94$\\ $wt = 18$ \end{tabular} }}			
				child { node[dotted] {\begin{tabular}{c} $|L_0^3| = 2^{109.49} $\\ $m' = 161$ \\ $wt = 32$ \end{tabular} }} 
				 }};
	\end{tikzpicture}}
	\caption{Optimization of~\autoref{fig:hgj-quantum} for $m = 255$. The size of the support is indicated by $m'$ and the weight by $wt$.}
	\label{fig:hgj-numerical}
\end{figure}

To compute the quantum time complexity, we consider the list sizes to be exact and use the formula of Lemma~\ref{lemma:quantum-merging} without the $O$. The subtrees on the right can be computed in $2^{65.71}$ operations; the slight increase is due to the fact that we take a sum of their respective terms and not a maximum. In the left branch, we sample from $L^0$ in $2^{63.48}$ operations.  

The actual time complexity is slightly bigger, due to the variation in list sizes, and the constant complexity overhead ($\pi/2$) of Grover search. More importantly, these operations require: $\bullet$~to recompute a sum, using $m$ (controlled) additions modulo $N$; $\bullet$~to test membership in some distribution; $\bullet$~to sample from input distributions $D^n$. The latter can be done using a circuit given in~\cite{ERBLM21}, which for a weight $k$ and $n$ bits, has a gate count $\OOt{ nk}$ and uses $n + 2\ceil{\log(k+1)}$ qubits. All of this boils down to $m$ arithmetic operations or $\OO{m^2}$ quantum gates.

Finally, this sampler works only for a proportion $\frac{1}{p_m} = 2^{-5.33}$ of subset-sum instances, so we need to re-randomize accordingly. After running the optimization for $128 \leq m \leq 1024$ and $n = m+1$, we found that the subset-sum solver would use approximately $2^{0.238m + 9.203}$ arithmetic operations, for a final list $L^0$ of size 2 on average. Under the subset-sum heuristic, we assume an independence between all tuples of elements in the initial lists. Using Chernoff-Hoeffding bounds the probability that the final list is empty is smaller than $e^{-1} \simeq 0.37$. To reduce it to a smaller constant $\varepsilon$, we may simply run multiple independent instances of the solver. This increases the asymptotic complexity by a factor $\OO{- \log \varepsilon }$.

\subsection{Solving Subset-Sum in Superposition}

We now show that we can reuse the structure of the QRACM-based subset-sum algorithm of~Figure~\ref{fig:hgj-quantum} to solve the problem in \emph{superposition} over the target $v$, while still keeping the number of qubits polynomial.

The basic idea is to reduce the problem with a given $v \neq 0$ to $v = 0$: $\innerprod{k'}{b} = 0 \mod N$, where $\vec{k'}$ is a length $m + 1$ vector where we append $-v$ to $\vec{k}$. We can then modify \emph{any} existing tree-based subset-sum algorithm solving this instance to force all vectors in the leftmost leaf list to have a 1 in the last coordinate, and all vectors in the other leaves to have 0 in this coordinate. Then only the lists in the left branch of the tree depend on $v$. The complexity is unchanged.

Following the tree structure of~\autoref{fig:hgj-quantum}, we create the lists $L_1^3, L_1^2, L_1^1$ in a precomputation step. Then we define a quantum algorithm that outputs an element in $L^0$ (or a superposition of such elements), and we run this algorithm in superposition over $v$.

\paragraph{Subset-Sum without QRACM.}
Helm and May~\cite{HM20} showed that quantum subset-sum algorithms using a small classical memory (without quantum access) can have better time-memory tradeoffs than classical ones. They obtained a time $\OOt{2^{0.428m}}$ for a memory $\OO{2^{0.285m}}$, however their algorithm does not have an unbalanced structure like ours.

We improve on this time-memory tradeoff by adapting the tree of~\autoref{fig:hgj-quantum} as follows: we remove $L_0^3$ and $L_1^3$ and their parameters, and directly sample in $L_0^2$. Assuming that the lists $L_1^2$ and $L_1^1$ are precomputed classically, we sample from $L^0$ with the same algorithm, except that it replaces each QRACM access (in time 1) by a sequential memory access (in time $|L_1^2|$ and $|L_1^1|$ for $L_1^2$ and $L_1^1$ respectively), $\ie$, a quantum circuit which encodes the elements of the lists as a sequence of standard gates. The asymptotic optimization gives a time $\OOt{2^{0.4165m}}$ with a memory $\OO{2^{0.2324m}}$. The parameters are displayed in~\autoref{fig:new-asymptotic}.

\begin{figure}[tb]
	\centering
	\scalebox{0.9}{\begin{tikzpicture}[grow=up,nodes={draw,rectangle,rounded corners=.25cm,->}, level 1/.style={sibling distance=75mm, level distance=18mm}, level 2/.style={sibling distance=35mm}, level 3/.style={sibling distance=28mm}]
	\node{ \begin{tabular}{c} $|L^0| = 2^0$ \end{tabular}}
	    child { node[very thick] {\begin{tabular}{c} $|L_1^1| = 2^{0.1669 m}$ \\ $c^1 = 0.4363m$ \\ $wt = 0.1474m$ \end{tabular} }
	    	child { node {\begin{tabular}{c} $|L_2^2| = |L_3^2| = 2^{0.2324m}$ \\ $c_1^2 = 0.1474m$ \\ $wt = 0.0737m$ \end{tabular} }
				child { node {\begin{tabular}{c} $|L_5^3| = 2^{0.1898m}$ \\$m' = m/2$ \\ $wt = 0.0737m/2$ \end{tabular} }}
				child { node {\begin{tabular}{c} $|L_4^3| = 2^{0.1898m} $ \\$m' = m/2$ \\ $wt = 0.0737m/2$ \end{tabular} }} }}
	    child { node[dotted] {\begin{tabular}{c} $|L_0^1| = 2^{0.4990m}$ \\ $c^1 = 0.4363m$ \\ $wt = 0.3527m$ \end{tabular} }
	    	child { node[very thick] {\begin{tabular}{c} $|L_1^2| = 2^{0.1439m}$ \\ $c_0^2 = 0.2878m$ \\ $wt = 0.0886m$ \end{tabular} }
	    		child { node {\begin{tabular}{c} $|L_3^3| = 2^{0.2158m} $\\$m' = m/2$ \\ $wt = 0.0886m/2 $ \end{tabular} }}
				child { node {\begin{tabular}{c} $|L_2^3| = 2^{0.2158m} $\\$m' = m/2$ \\ $wt =0.0886m/2 $ \end{tabular} }} 
				 }
	    	child { node[dotted] {\begin{tabular}{c} $|L_0^2| = 2^{0.5450m} $ \\ $c_0^2 = 0.2878m$ \\ $wt = 0.2641m$ \end{tabular} }
				 }};
	\end{tikzpicture}}
	\caption{Asymptotic optimization of our quantum subset-sum algorithm without QRACM. The lists on the right of the tree are constructed with classical computations, using classical RAM. The lists $L_1^2$ and $L_1^1$ are stored in classical memory without random access.}
	\label{fig:new-asymptotic}
\end{figure}
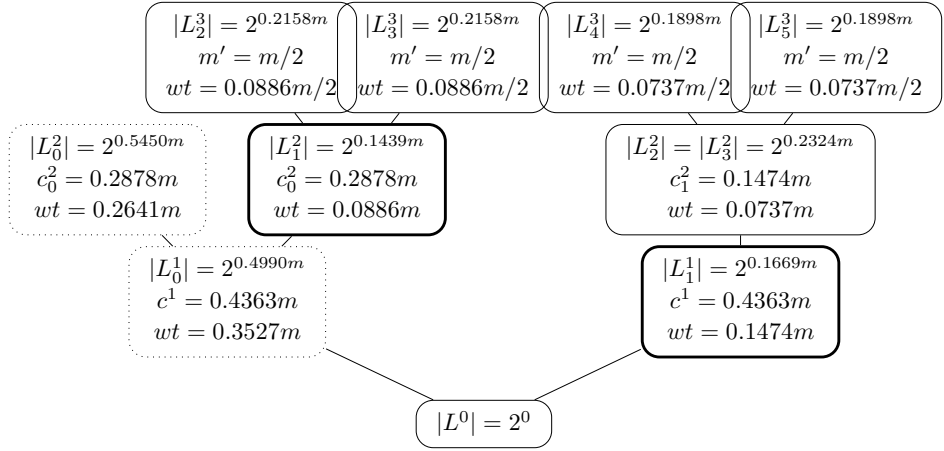

The difference between asymptotic and non-asymptotic optimization is bigger here. For $m = n-1 = 127$, with the constraint $|L^0| = 2^2$, we obtain a time $2^{60.01} > 2^{128 \times 0.4165} = 2^{53.31}$ and a memory $2^{26.82} < 2^{128 \times 0.2324} = 2^{29.75}$. On top of this, we must also take $p_m$ into account.

After running optimizations for $n = 128$ to $1024$, we obtained a count of about $2^{0.418m + 12.851}$ blocks of $m$ arithmetic operations ($m^2$ quantum gates). The point at which the algorithm starts improving over Grover search lies around $n = 157$.

\subsubsection*{Acknowledgments.}
This work has been partially supported by the French Agence Nationale de la Recherche through the France 2030 program under grant agreement No. ANR-22-PETQ-0008 PQ-TLS.

\newcommand{\etalchar}[1]{$^{#1}$}

\end{document}